\documentclass[journal, 10pt]{resources/IEEEtran}
\IEEEoverridecommandlockouts                 

\usepackage[utf8]{inputenc}

\usepackage{amsthm, amsmath, amsfonts, amssymb}
\usepackage{resources/custom_commands}
\usepackage{resources/coloredboxes}

\usepackage{enumerate}
\usepackage{graphicx}
\usepackage{mathtools}
\usepackage{thmtools}
\usepackage{thm-restate}
\usepackage{float}
\usepackage{enumerate}
\usepackage{marginnote}
\usepackage[hidelinks]{hyperref}
\usepackage{cleveref}
\usepackage{subfigure}

\usepackage[shortlabels]{enumitem}

\usepackage{autonum}

\def\eqdef{\stackrel{\triangle}{=}}

\newcommand{\eqlinebreakshort}{\ensuremath{\nonumber \\ & \quad \quad}}

\newcommand{\degree}{\ensuremath{\mathsf{deg}}}
\newcommand{\diag}{\ensuremath{\mathsf{diag}}}

\newcommand{\fake}{d}

\newcommand{\dist}{\text{dist}}
\newcommand{\error}{\text{err}}

\newcommand{\dermat}[1]{\ensuremath{\bJ_{#1}}}

\newcommand{\revise}[1]{\textcolor{black}{#1}}
\newcommand{\revisered}[1]{\textcolor{black}{#1}}
\newcommand{\reviseblue}[1]{\textcolor{black}{#1}}

\newif\ifeight
\eighttrue

\newif\ifshowproof
\showprooftrue

\newif\ifsixteen

\ifeight
    \sixteentrue
\else
    \sixteenfalse
\fi

\begin{document}

\title{\LARGE \bf
Stochastic Opinion Dynamics under Social Pressure in Arbitrary Networks 
{}
\thanks{The work was supported by ARO MURI W911 NF-19-1-0217 and a Vannevar Bush Fellowship from the Office of the Under Secretary of Defense.
}
}

\author{Jennifer Tang, Aviv Adler, Amir Ajorlou, and Ali Jadbabaie
}
\date{\today}

\maketitle

\thispagestyle{plain}
\pagestyle{plain}


\begin{abstract}
    Social pressure is a key factor affecting the evolution of opinions on networks in many types of settings, pushing people to conform to their neighbors' opinions. To study this, 
    the \emph{interacting P\'{o}lya urn} model was introduced by Jadbabaie et al.~\cite{socialPressure2021}, in which 
    each agent has two kinds of opinion: \emph{inherent beliefs}, which are hidden from the other agents and fixed; and \emph{declared opinions}, which are randomly sampled at each step from a distribution which depends on the agent's inherent belief and her neighbors' past declared opinions (the social pressure component), and which is then communicated to her neighbors.
    Each agent also has a \emph{bias parameter} denoting her level of resistance to social pressure. 
    At every step, each agent updates her declared opinion (simultaneously with all other agents) according to her neighbors' aggregate past declared opinions, her inherent belief, and her bias parameter. 
    We study the asymptotic behavior of this opinion dynamics model and show that the agents' declaration probabilities \revisered{approaches a set of equilibrium points of the expected dynamics} using Lyapunov theory and stochastic approximation techniques. We also derive necessary and sufficient conditions for the agents to approach consensus on their declared opinions. Our work provides further insight into the difficulty of inferring the inherent beliefs of agents when they are under social pressure.
\end{abstract}


\section{Introduction and Related Work}

Opinion dynamics -- the modeling and study of how people's opinions change in a social setting (particularly through communication on a network, whether online or offline) -- is an extremely useful tool for analyzing various social and political phenomena such as consensus and social learning \cite{Noorazar2020} as well as for designing strategies for political, marketing and information campaigns, such as the effort to curb vaccine hesitancy \cite{Ancona2022}. It is generally assumed in such models that the agents report their opinions truthfully. 
In reality, however, there are many occasions in which people make declarations contrary to their real views in order to conform socially \cite{Milosz1953}, a fact confirmed both by common sense and by psychological studies \cite{Asch1956}. This can make it difficult to determine the true beliefs governing observed interactions. 



In this work, we study an \emph{interacting P\'{o}lya urn model} for opinion dynamics, originating from \cite{socialPressure2021}, that captures a system of agents who might be untruthful due to their local social interactions. 
\ifeight 
\fi 
The goal of this model is to shed light on how opinions might evolve in the presence of social pressure. Jadbabaie et al. \cite{socialPressure2021} considered whether it is possible to infer a person's true or inherent belief from their declared opinions under this model, specifically studying an aggregate estimator on a complete graph network.

\ifeight 

\revise{We refer the reader to \cite{socialPressure2021} and \cite{opiniondynamicsCDC} for in-depth discussion of prior work. Here, we discuss some relevant highlights. Influential agent-based opinion dynamics models include the DeGroot model \cite{DeGroot1974}
and the Friedkin-Johnsen model \cite{friedkin1990}. In the DeGroot model, under reasonable network assumptions, all agents eventually agree on the same opinion (represented by a real number), a state called consensus. In the Friedkin-Johnsen model, agents continue to update their opinion with their initial opinion, causing consensus to never occur. Works which also try to model lack of agreement among agent include \cite{friedkin1990,HegselmannKrause2002,Deffuant2002}. 
Other similar models consider competitive contagion and product adoption \cite{kempe2003maximizing, bharathi2007competitive, fazeli2017competitive, amini2009marketing}. Stochastic models which focus on binary opinions include the voter model 
\cite{holleyLiggett1975,  Yildiz2013stubborn}.
Other models where agents have a separate internal and external opinion include \cite{centola_2005} and \cite{Ye2019}. 
Those looking at opinion dynamics under social pressure include \cite{Das2014, amelkin2017, semonsen2018opinion, ferraioli2019social, cheng2019opinion,  liu2021modeling} though there are many others.}

\fi 


\subsection{Contributions}

While \cite{socialPressure2021} originally proposed an interacting P\'{o}lya urn model for opinion dynamics, they studied it only in the special case of a (unweighted) complete graph as the network, with agents that all have the same honesty parameter $\tilde \gamma$. 
In this work we remove these constraints and study this process on arbitrary undirected graphs with agents whose honesty parameters may differ. 
Our contributions are:
\begin{enumerate}
    
    \item We establish that the behavior of agents (i.e. their probabilities of declaring each opinion) almost surely converges to \revisered{a set of equilibria}.
    \item We determine necessary and sufficient conditions for \emph{consensus}. 
    %
   Due to the stochastic nature of our model we define consensus as a property of the \emph{declared} opinions: the network \emph{approaches consensus} if all agents declare the same opinion (either $0$ or $1$) with probability approaching $1$ as time goes to infinity. This corresponds to cases where social pressure forces increasing conformity over time, and makes estimating the agents' inherent beliefs from their behavior difficult, as shown in \cite{socialPressure2021}.
\end{enumerate}

We discuss more details of these contributions and their comparison with \cite{socialPressure2021} below.

\paragraph*{Behavior of Agent Declaration Probabilities}

\revise{
Previously in \cite{socialPressure2021}, the authors determined that the dynamics of the interacting P\'{o}lya urn model converge to a fixed point and identified this fixed point. However, since their work was restricted to the complete graph, showing this only involved understanding the evolution of a $1$-dimensional variable. For arbitrary networks, it is necessary to characterize the dynamics of an $n$-dimesional process.
}

We use Lyapunov theory and stochastic approximation to determine \reviseblue{ the convergence set} for the opinion dynamics model \revise{when it is a more complicated $n$-dimesional process}. We show that on \revise{arbitrary} undirected networks, the probability that each agent $i$ declares $1$ at the next step converges asymptotically to \reviseblue{some set of points each representing an equilibrium of the expected dynamics of the process.} 

\reviseblue{While we show that the process must converge to a set of equilibrium points, it is unknown whether the process has a unique stable equilibrium point or multiple stable equilibria,} \revisered{ or cycles within a set of connected equilibria.}


%

\paragraph*{Conditions for Consensus}

An interesting result from \cite{socialPressure2021} is that if the proportion of agents (connected in a complete graph) with inherent beliefs $0$ or $1$ passes a certain threshold, then asymptotically the system almost surely converges to \revise{a state where all agents declare the same opinion with probability $1$}. In this work, we find an analogous result for general networks, determining a condition under which all agents in the network almost surely converge to consensus, \revise{a state where the dynamics converge to either the all-zero equilibrium point or the all-one equilbrium point}. The condition is derived by incorporating the structure of the network, the inherent beliefs of the agents and their honesty \reviseblue{(or bias)} parameters. 

\ifeight 
\fi 

Our contributions give insight on the difficulty of inferring inherent beliefs of agents in the network, a key question explored in \cite{socialPressure2021}, which specifically determined whether it is possible to infer the inherent beliefs of agents from the history of declared opinions. 
\revise{The results were that this estimator may have difficulties in estimating the inherent belief of all agents (even in the limit) if the network approaches consensus. 
How to estimate the honesty parameter (or, equivalently, the bias parameter) of any agent in any network was further studied in \cite{opinionDynamicsPart2}, where it was determined that approaching consensus presented the most difficult case.}

\section{Model Description}


Our model is a slight generalization of the model from \cite{socialPressure2021}, with the addition that each edge in the network has a (nonnegative) weight denoting how much the two agents' declared opinions influence each other. We introduce this generalization as it does not affect our results%
\ifeight 
.
\fi 
As mentioned, we also extend the model by permitting agents to have different honesty parameters \revise{(quantities we replace with bias parameters)}.

\subsection{Graph Notation}

Let (undirected) graph $G = (V,E)$ be a network of $n$ agents (the vertices) 
labeled $i = 1, 2, \dots, n$, so $V = [n]$. The graph $G$ can have self-loops. For each edge $(i,j) \in E$, there is a weight $a_{i,j} \geq 0$, where by convention we let $a_{i,j} = 0$ if $(i,j) \not \in E$. We denote the matrix of these weights as $\bA \in \bbR^{n \times n}$, i.e. the weighted adjacency matrix of $G$; since $G$ is undirected, $\bA$ is symmetric.



The vector of degrees of all agents is denoted as
$
    \bd \eqdef [\degree(1), \degree(2), \dots, \degree(n)]
$ 
and its diagonalization is denoted $\bD = \diag(\bd)$, i.e. the diagonal matrix of the degrees. Let the \emph{normalized adjacency matrix} be 
$
    \bW = \bD^{-1} \bA \,. \label{eq::graph_matrix}
$
We assume that $\bW$ is irreducible ($G$ is connected) and not bipartite.

We use $\bI$ as the identity matrix. 
We denote the largest eigenvalue of a matrix by $\lambda_{\max}(\cdot)$ (the matrices we use this with have real eigenvalues), and the indicator function by $\bbI\{ \cdot\}$. We denote an all-$0$ vector as $\bzero$ and an all-$1$ vector as $\bone$.

\subsection{Inherent Beliefs and Declared Opinions}

Each agent $i$ has an \emph{inherent belief} $\phi_i\in \{0,1\}$, which does not change.
%
At each time step 
$t$,
each agent $i$ (simultaneously) announces a \emph{declared opinion} $\psi_{i,t} \in \{0,1\}$. 
The declarations $\psi_{i,t}$ are based on a probabilistic rule (which we will define shortly) and is determined by the declared opinions the agents sees in her neighborhood and her \emph{bias parameter}. 
\revise{The bias parameter, internal to agents, measures an agent's inclination to opinion $1$.}
\revise{The bias parameter $\gamma_i$ combines the inherent belief $\phi_i$ and honesty parameter\footnote{We permit heterogeneous honesty parameters, while $\tilde \gamma_1= \ldots=\tilde \gamma_n=\gamma$ in \cite{socialPressure2021}).} $\tilde \gamma_i$ of agent $i$ used in \cite{socialPressure2021} into a single parameter where $\gamma_i = \tilde \gamma_i $ if $\phi_i = 1$ and  $\gamma_i = 1/\tilde \gamma_i $ if $\phi_i = 0$.}
\revise{
In terms of the bias parameter, we can write inherent beliefs as
\begin{align}
    \phi_i = \begin{cases} 1 &\text{if } \gamma_i > 1 \\ 0 &\text{if } \gamma_i < 1 \end{cases}
    \,.
\end{align}
An agent is neutral (has no inherent opinion) if $\gamma_i = 1$.
}
 
We used $\bgamma = [\gamma_1, \dots, \gamma_n]$ to specify the bias parameters of all agents. 
Next, for $t \in \bbZ_{+}$ let
\begin{align}
 M_{i}^0(t) &= m^0_i + \sum_{\tau = 2} ^{t} \sum_{j = 1 }^n a_{i,j} \bbI[\psi_{j, \tau} = 0] \label{eq::M0_counts_def}\\
    M_{i}^1(t) &= m^1_i+ \sum_{\tau = 2} ^{t}\sum_{j=1}^n a_{i,j}  \bbI[\psi_{j, \tau} = 1]\label{eq::M1_counts_def}
\end{align}
where $m^0_i, m^1_i > 0$ represent the initial settings of the model. 
(Initial settings are used in place of declared opinions at time $1$. Some requirements for the initial settings are given shortly.)
The quantity $M_{i}^0(t)$ represents the (weighted) number of times agent $i$ observed a neighbor declare opinion $0$ up to step $t$ (plus initial settings), and $M_{i}^1(t)$ represents the analogous total of observed $1$'s. 
%
%
%
We will normalize these values using the following: let $M_{i}(t) \eqdef m^0_i + m^1_i + (t-1) \, \degree(i) =  M_{i}^0(t) + M_{i}^1(t) $ and then set
$\mu_{i}^0(t) \eqdef  M_{i}^0(t) / M_{i}(t)$ and $ \mu_{i}^1(t) \eqdef  M_{i}^1(t) / M_{i}(t) \,.
$
Since $\mu_i^0(t) = 1 - \mu_i^1(t)$, we simplify the notation to
$
    \mu_i(t) \eqdef \mu_i^1(t)\,.
$



Define the function (note that $\mu, \gamma$ are scalars)
\begin{align}\label{eq::f_def}
f(\mu, \gamma) \eqdef \frac{\gamma \mu}{1 + (\gamma - 1) \mu} = \frac{1}{1+ \frac{1}{\gamma}\left(\frac{1}{\mu} - 1 \right)}
\end{align}
and then let the probability of declared opinions be
\begin{align}\label{eq::declared_opinion_prob_2}
    \psi_{i,t+1} \eqdef \left\{\begin{array}{lll}
    1 & \text{ with probability } p_i(t) =  f(\mu_i(t), \gamma_i)  \\
    0 & \text{ with probability } 1- f(\mu_i(t), \gamma_i) 
    \end{array}\right.\,.
\end{align}
Function $f(\mu\, \gamma)$ is increasing in $\mu$ and thus when more of agent $i$'s neighbors declare $1$, this increases the probability of agent $i$ declaring $1$. The bias parameters is a multiplicative factor which scales the number of $1$ observations.

Note that the bias parameter $\gamma_i$ is always defined as agent $i$'s bias towards opinion $1$. However, the model is symmetric in the following way: a $\gamma$ bias towards $1$ is equivalent to a $1/\gamma$ bias towards $0$, which is captured by the equation
\begin{align}
\label{eq::reciprocal_for_f}
    f(\mu_i^1(t), \gamma) = 1 - f(\mu_i^0(t), 1/\gamma)\,.
\end{align}

We also define a normalized value that summarizes agent $i$'s declarations. 
Let $b_i^0, b_i^1 >0$ (the initialization) be such that
$b_i^0 + b_i^1 = 1$ for each $i$ and
\begin{align}
    m^0_i &= \sum_{j = 1}^n a_{i,j} b_{j}^0
    ~\text{ and }~ m^1_i = \sum_{j=1}^n  
 a_{i,j}  b_{j}^1\,.
\end{align}
For $t \in \bbZ_{+}$, let
\begin{align}
    \beta_{i}^0(t) &= \frac{b_{i}^0}{t} + \frac{1}{t}\sum_{\tau = {2}}^t (1-\psi_{i,\tau}) , \,\,\,
    \beta_{i}^1(t) 
    =\frac{b_{i}^1}{t} +\frac{1}{t}\sum_{\tau = {2}}^t \psi_{i,\tau}\,.
\end{align}

These are counts and proportions of declarations of each opinion (or ``time-averaged declarations'') for each agent (plus initial conditions). 
We similarly use  
$\beta_{i}(t) \eqdef \beta_{i}^1(t)$. It then follows that
\begin{align}
    \mu_i(t) =  \frac{1}{\degree(i)}\sum_{j = 1}^n a_{i,j}\beta_j(t)\,.
\end{align}

Finally, we let the corresponding vectors at time $t$ be
\begin{align}
    \bmu(t) &\eqdef \left[\mu_1(t),...\mu_n(t)\right]^{\ltop}
    \text{, }
    \bbeta(t) \eqdef \left[\beta_1(t),...\beta_n(t)\right]^{\ltop}\,
    \\
    &  \text{and } \bpsi(t) \eqdef [\psi_{1, t}, \dots, \psi_{n, t}]^{\ltop}
\end{align}

Define the diagonal matrix with $\bgamma$ along the diagonal as
\begin{align}
\bGamma = \diag(\bgamma)\label{eq::diag_gamma_matrix}\,.
\end{align}
We assume for this work that $\bGamma \neq \bI$. This parallels the assumption $\tilde \gamma_i > 1$ used in \cite{socialPressure2021}. 


\revise{Let $\cH_t$ be the \emph{history} of the process, consisting of all $\psi_{i,\tau}$ for $\tau \leq t$, and thus contains all information declared up to and including time $t$. The same notation, $\cH_0, \cH_1, \dots$ will be used as a filtration for the stochastic process generated by these dynamics, where $\cH_t$ is the $\sigma$-algebra for all events that happen by the time $t$.}

%
The recursive equation that governs the count of declared opinions by agent $i$ is
\begin{align}
    \beta_{i}(t+1) 
    = \frac{t}{t+1} \beta_{i}(t) + \frac{1}{t+1} \psi_{i, t+1}.\label{eq::evolution_stochastic}
\end{align}

\ifeight 
Conditioned on this history and combining all agents together, this gives the expectation

\fi 
\begin{align}
     \bbE[ \bbeta(t+1)-\bbeta(t)|\cH_{t}] &= \frac{1}{t+1}(f (\bmu(t), \bgamma)- \bbeta(t)).
\end{align}

\begin{definition}\label{def::det_expected_dynamics}
The \emph{deterministic expected dynamics} are
\begin{align}
    \bbeta(t+1)-\bbeta(t) &= \frac{1}{t+1}(F(\bbeta(t), \bgamma)-\bbeta(t))
\end{align}
where $F(\bbeta(t), \bgamma) = [F_1(\bbeta(t), \bgamma),\dots,F_n(\bbeta(t), \bgamma)]$ and
\begin{align}
    F_i(\bbeta(t), \bgamma) = f\left(\frac{1}{\degree(i)}\sum_{j =1}^n a_{i,j} \beta_j(t), \gamma_i\right)\,.
\end{align}
\end{definition}
We refer to original dynamics governed by \eqref{eq::declared_opinion_prob_2} and \eqref{eq::evolution_stochastic} as
the \emph{full stochastic dynamics}.





\ifeight 
\fi 

 \section{Convergence Analysis}

\revise{ 
One of the primary contributions of the present work is to show the behavior of the time-averaged declared opinions $\bbeta(t)$ of the interacting P\'{o}lya urn model under the stochastic dynamics of  \eqref{eq::declared_opinion_prob_2} and \eqref{eq::evolution_stochastic} in arbitrary networks. Since $\bbeta(t)$ approaches equilibrium points, we first define equilibrium points and proceed to developing tools for this result.}

\subsection{Equilibria of the Expected Dynamics}
\begin{definition}
A vector $\bbeta$ is an \emph{equilibrium point} of the expected dynamics if
$
    F(\bbeta, \bgamma) = \bbeta\,.
$
\end{definition}

While this is defined in terms of $F(\bbeta, \bgamma)$ which is for the deterministic expected dynamics, the same equilibrium points are equilibrium points for the full stochastic dynamics. 

Note that vector $\bone$ and vector $\bzero$ are always equilibrium points. \revise{We discuss this further in \Cref{sec::converge_to_consensus}.}

Finding an exact analytic expression for equilibrium points for a given network with given bias parameters is unfortunately difficult in general. 
We present a set of equations which can be used numerically to solve for equilibrium points. 

\begin{proposition}\label{prop::equilibrium_beta}
The equilibrium points of the expected dynamics 
are given by the solutions to the equations
\begin{align}\label{eq::equilibrium_beta}
    0 = (\gamma_i - 1)\beta_i\mu_i 
    + \beta_i - \gamma_i \mu_i  
\end{align}
where $i \in \{1,\ldots, n\}$ and $\mu_i = \frac{1}{\degree(i)}\sum_{j=1}^n a_{i,j}\beta_j$.

\end{proposition}

\begin{proof}
This follows from the fact that at any equilibrium,
\begin{align}
    \beta_i &= f(\mu_i,\gamma_i) = \frac{\gamma_i \mu_i }{1 + (\gamma_i - 1) \mu_i}.
\end{align}
\end{proof}





\subsection{Tools from Stochastic Approximation}

\revise{
In order to prove that the full stochastic dynamics approaches the equilibrium points of the expected dynamics, we use results based on the long-term behavior of path-dependent stochastic processes presented in \cite[Theorem 3.1]{arthur1986}, where stochastic approximation is used to show convergence of general  P\'{o}lya urn models.  
}

\revise{
Recall that the full stochastic dynamics of our model (captured by \eqref{eq::evolution_stochastic} and \eqref{eq::declared_opinion_prob_2}) is (for each $i$):
\begin{align}
    \beta_{i}(t+1) 
    &= \beta_{i}(t) + \frac{1}{t+1} \big(\psi_{i, t+1} - \beta_{i}(t)\big) \label{eq::urn_dynamics_full_psi}\,.
\end{align}
General P\'{o}lya urn models follow the above equation as well, where the random draws $\psi_{i, t+1}$ are governed by \emph{urn functions} 
${\boldf} = (f_{1}, \dots, f_{n})$
where
\begin{align}
\label{eq::equation_general_polya_urn1}
    \psi_{i, t} = \begin{cases}
    1 & \text{ with probability }f_{i}(\bbeta(t)) \\
    0 & \text{ with probability }1 - f_{i}(\bbeta(t)) 
    \end{cases}
\end{align}
for all $i=1,\ldots,n$.
In our work, we require that  $\bbeta(t) \in [0,1]^n$ and our urn functions 
are given in \eqref{eq::declared_opinion_prob_2}.
For any general set of urn functions, we can define (similarly to \Cref{def::det_expected_dynamics}) the corresponding deterministic expected dynamics as (for each $i$)
\begin{align}
\beta_i(t+1) - \beta_i(t) = \frac{1}{t+1} (f_i(\bbeta(t)) - \beta_i(t))\,.
\end{align}
}

\revise{
Let $\cS$ be a subset of $[0, 1]^n$ (and we also allow that $\cS = [0, 1]^n$). Let the set 
\begin{align}
    \label{eq::general_eq_points}
    B_{\boldf} = \{ \bbeta: \boldf(\bbeta) = \bbeta, \bbeta \in \cS\}
\end{align}
denote the set of equilibrium points of the deterministic expected dynamics with urn function $\boldf$. 
Then, we can show that the dynamics of general urn functions gets arbitrarily close to the equilibrium points of the expected dynamics. For this theorem, we need the following notation: For any set $A$, let $\dist(\bbeta, A) = \inf_{\ba \in A} \norm{\bbeta - \ba}$. 
}

\revisered{
\begin{theorem}[Adapted from Theorem 3.1 from \cite{arthur1986}]\label{thm::arthur1986}
Let $\cS$ be a compact and convex subset of $[0, 1]^n$ and $\boldf = (f_1, \dots, f_n)$ be continuous urn functions $\boldf: \cS \to \cS$ governing 
the general P\'{o}lya urn model whose dynamics $\bbeta(t) \in \cS$ evolve according to
\eqref{eq::urn_dynamics_full_psi}
and \eqref{eq::equation_general_polya_urn1}. 
If there exists a Lyapunov function $V: \cS \to \bbR$ where  \label{part::V_arthur1986}
    \begin{enumerate}[(i)]
    \item $V$ is twice \reviseblue{continuously} differentiable
    \item $V(\bbeta) \geq 0$ for $\bbeta \in \cS$
    \item \label{part::V_iii_arthur1986}
    $ \langle \boldf(\bbeta) - \bbeta, \nabla V(\bbeta)\rangle < 0$ for $\bbeta \in \cS \setminus B_{\boldf}$ and equal to $0$ for $\bbeta \in B_{\boldf}$ 
    \end{enumerate}
then $\lim_{t \to \infty} \dist(\bbeta(t), B_{\boldf}) = 0$ almost surely.
\end{theorem}
}
\revise{
In other words, if the conditions are met, the stochastic dynamics of the general urn model approaches the set of equilibrium point of the deterministic expected dynamics. If the set of equilibrium points consist only of distinct points, then the above theorem shows that the dynamics must converge to an equilibrium point of the expected dynamics. }

\ifshowproof 
\fi 

\ifshowproof 
\reviseblue{We remark that the condition where $V$ needs to be \emph{continuously} twice differentiable is not strictly necessary. The reason behind this is that in our proof we need the Hessian to always be bounded. Since we are also working on a compact set, continuous Hessians imply that the Hessian is bounded. We could have instead stated the theorem as $V$ having bounded Hessian.}

\revise{\Cref{thm::arthur1986} states that the key quantity needed to be shown is the condition on the inner product in \ref{part::V_iii_arthur1986}, that the direction of movement of any one step of the dynamics to its expected value is always in a direction that decreases the Lyapunov function.}

\revise{
We include the proof of \Cref{thm::arthur1986} for completeness. First, we include the following lemma about concerning martingales.
}

\revisered{
\begin{lemma} \label{lem:martingale-decreasing-steps}
    If $Y_1, Y_2, \dots$ is a martingale with $|Y_1| < c$ and step sizes decreasing according to $|Y_t - Y_{t-1}| \leq c/t$ for some constant $c > 0$, then $\lim_{t \to \infty} Y_t$ exists and is finite almost surely.
\end{lemma}
}

\begin{proof}
We want to use Doob's Martingale Convergence theorem, which requires showing that $\sup_t \bbE[|Y_{t}|]$ is finite. We will bound this quantity using the variance of $Y_{t}$.
Since
\begin{align}
Y_{t-1} - c/t \leq Y_t \leq Y_{t-1} + c/t
\end{align}
this implies that
\begin{align}
\var[Y_t | Y_{t-1}] \leq \left(\frac{c}{t} \right)^2
\end{align}
and we also get that $\var[Y_1] < c^2$.

For any $t >0$, we apply the law of total variance and telescope the resulting sum to get
\begin{align}
\var[Y_{t}] &= \bbE[\var[Y_{t}|Y_{t-1}]] + \var[\bbE[Y_{t}|Y_{t-1}]]
\\ &\leq \left(\frac{c}{t}\right)^2 + \var[Y_{t-1}]
\\ \implies & \var[Y_{t}] \leq \var[Y_1] +\sum_{\tau=2}^t \left(\frac{c}{\tau}\right)^2  \leq \sum_{\tau=1}^t \left(\frac{c}{\tau}\right)^2 \,.
\end{align}
 But since and $\bbE[|Y_t|] \leq \bbE[Y_t^2]^{1/2}$ (using Cauchy-Schwarz inequality) and $\bbE[Y_t] = 0$ (WLOG we can let $Y_1 = 0$),
\begin{align}
\sup_{t > 0} \bbE[|Y_{t}|] &\leq \sup_{t > 0} \bbE[Y_{t}^2]^{1/2}
\\ &\leq \sup_{t > 0} (\var(Y_t) + \bbE[Y_t]^2)^{1/2}
\\ & \leq \left(\sum_{\tau=1}^\infty \left(\frac{c}{\tau}\right)^2 \right)^{1/2} < \infty
\end{align}
Thus, we can apply Doob's Martingale Convergence Theorem to show that $Y_{t}$ almost surely converges to a finite limit.
\end{proof}

\begin{proof}[Proof of \Cref{thm::arthur1986}] For notation let $s$ be the diameter of  set $\cS$.  Furthermore, note that $\langle \boldf(\bbeta) - \bbeta, \nabla V(\bbeta) \rangle$ is continuous since $\boldf(\bbeta)$ is continuous and $\nabla V(\bbeta)$ is differentiable (since $V$ is twice differentiable) and thus also continuous.


For notation let:
\begin{align}
\Delta(\bbeta(t)) \eqdef \bpsi(t+1) - \bbeta(t)
\end{align}
and for any $\bbeta \in \cS$,
\begin{align}
g(\bbeta) = \bbE[\Delta(\bbeta)] = \boldf(\bbeta) - \bbeta
\end{align}

Our procedure for showing our result is to approximate the step-wise change in Lyapunov function $V$ by its linear approximation using $\nabla V$ and the value at the \emph{expected step}, that is, we want to show that
\begin{align}
\label{eq::linear_approximation_expression}
V (\bbeta(t+1)) - V (\bbeta(t)) \approx
\langle  \nabla V(\bbeta(t)), g(\bbeta(t)) \rangle
\end{align}
Then, we argue that if the dynamics $\bbeta(t)$ does not get arbitrarily close to the set $B_{\boldf}$, we will get that
\begin{align}
\lim_{t \to \infty}  \sum_{\tau > t_0}^{t} \langle  \nabla V(\bbeta(\tau)), g(\bbeta(\tau)) \rangle = - \infty
\end{align}
which implies also that
\begin{align}
\lim_{t \to \infty}  \sum_{\tau > t_0}^{t} V (\bbeta(\tau+1)) - V (\bbeta(\tau)) = - \infty
\end{align}
However, this is a contradiction, since the $V > 0$ and $V(\bbeta(0))$ starts a finite value, the total change in the Lyapunov function over the dynamics cannot be infinitely negative. 

To make this argument precise, first, there are two types of errors which occur from approximation \eqref{eq::linear_approximation_expression} which we need to show are small. The first is the error of using the linear approximation. The second is the error of using the expected step $g(\bbeta)$ in place of $\Delta(\bbeta(t))$.
We start with the first error. 

Define the error of of using the linear approximation as
\begin{align}
\error(\bbeta,\bz) \eqdef V(\bz) - (V(\bbeta) + \langle \nabla V(\bbeta), \bz - \bbeta \rangle)\,.
\end{align}
Since we are working on a compact set $\cS$ where $V$ is twice continuously differentiable, this means that $|\nabla^2 V(\bbeta)| < h $ for some constant $h > 0$ which holds for all $\bbeta$. Using \cite[Lemma 1.2.2 and Lemma 1.2.3]{nesterov2018lectures}
this means that
\begin{align}
 V(\bz) - (V(\bbeta) + \langle \nabla V(\bbeta), \bz - \bbeta \rangle) \leq \frac{h}{2} \norm{\bz - \bbeta}^2
\end{align}
and thus for $c_1 = h/2$ we have
\begin{align}
\label{eq::err_bound_square_norm}
|\error(\bbeta,\bz)| \leq c_1 \norm{\bz - \bbeta}^2 \,.
\end{align}

Since $\bbeta(\tau + 1) - \bbeta(\tau) = \frac{1}{\tau + 1}\Delta(\bbeta(t))$, using \eqref{eq::err_bound_square_norm},
\begin{align}
&\left|\sum_{\tau = t_0}^{t-1} \error\left(\bbeta(\tau),\bbeta(\tau+1)\right) \right|
\\&\leq \left|\sum_{\tau = t_0}^{t-1} c_1 \frac{\norm{\Delta(\bbeta(t))}^2}{(\tau + 1)^2} \right|
\\& \leq  c_1 \sum_{\tau = t_0}^{t-1}\frac{s^2}{(\tau + 1)^2} 
\\ &=  K_1< \infty  
\end{align}
for some constant $K_1 > 0$.

For the second error, we want to show that the perturbations from the expectation converge. Let 
\begin{align}
\boldeta(t) 
\eqdef \Delta(\bbeta(t)) - \bbE[\Delta(\bbeta(t))] 
\end{align}
and define
\begin{align}
y(t) = \sum_{\tau=1}^t \frac{1}{\tau} 
\langle \nabla V(\bbeta(\tau-1)), \boldeta(\tau) \rangle\,.
\end{align}
Since  $\nabla V$ has bounded magnitude (as it is continuous on the compact set $\cS$ since $V$ is continuously differentiable) and 
\begin{align}
\bbE[\boldeta(\tau)|\bbeta(\tau-1)] = \bzero
\end{align}
we get that 
    \begin{align}
        \bbE&[y(t) - y(t-1) \,|\, \bbeta(t-1)] 
        \\&= \frac{1}{t} \langle\nabla V(\bbeta(t-1)), \bbE[\boldeta(t) \,|\, \bbeta(t-1)] \rangle = \bzero
    \end{align}
    and hence $y(t)$ is a martingale. Furthermore, since $|\langle \nabla V(\bbeta(t-1)), \boldeta(t)\rangle |$ is bounded, the martingale $y(1), y(2), \dots$ has step sizes of $c_2 / t$ for some constant $c_2$. This means that by \Cref{lem:martingale-decreasing-steps}, it converges to a (finite) limit almost surely, which additionally implies that $\sup_t |y(t)| = K_2$ for some constant $0 < K_2 < \infty$ almost surely as well.

Next, we express the difference of the value of $V$ under our dynamics. Let $t$ be some time larger than time $t_0$. Then
\begin{align}
V &(\bbeta(t)) - V(\bbeta(t_0)) = \sum_{\tau = t_0} ^{t-1} V(\bbeta(\tau + 1)) - V(\bbeta(\tau ))
\\ & =  \sum_{\tau = t_0} ^{t-1} \left\langle  \nabla V(\bbeta(\tau)), \frac{1}{\tau + 1} \Delta(\bbeta(\tau)) \right\rangle  
\eqlinebreakshort
+ \error\left(\bbeta(\tau), \bbeta(\tau+1)\right) 
\\ & =  \sum_{\tau = t_0} ^{t-1} \left\langle  \nabla V(\bbeta(\tau)), \frac{1}{\tau + 1} g(\bbeta(\tau) ) \right\rangle  
\eqlinebreakshort
+\left\langle  \nabla V(\bbeta(\tau)), \frac{1}{\tau + 1} \boldeta(\bbeta(t)) \right\rangle  
\eqlinebreakshort
+ \error\left(\bbeta(\tau), \bbeta(\tau+1)\right) 
\end{align}

Since $V > 0$, it must be that $V (\bbeta(t)) - V(\bbeta(t_0)) $ is finite which gives
\begin{align}
\sum_{\tau = t_0} ^{t-1} &\left\langle  \nabla V(\bbeta(\tau)), \frac{1}{\tau + 1} g(\bbeta(\tau)) \right\rangle
\\&> V (\bbeta(t)) - V(\bbeta(t_0)) - K_1 - K_2
\end{align}
and therefore
\begin{align}
\label{eq:finite-drift}
\lim_{t \to \infty}\sum_{\tau = t_0} ^{t-1}  \frac{1}{\tau + 1} \left\langle  \nabla V(\bbeta(\tau)), g(\bbeta(\tau)) \right\rangle  > - \infty\,.
\end{align}

We will use the above equation to derive a contradiction. Assume that $\lim_{t \to \infty} \dist(\bbeta(t), B_{\boldf}) \neq 0$ (or does not converge to a limit at all). 
Let $U(X, \epsilon)$ be the $\epsilon$ open neighborhood of set $X$.
Select $\epsilon > 0$. Let $C(\epsilon) = \cS \setminus  U(B_{\boldf}, \epsilon)$, which are all points bounded away from $B_{\boldf}$ by $\epsilon$; note that $C(\epsilon)$ is also compact.

The statement $\lim_{t \to \infty} \dist(\bbeta(t), B_{\boldf}) \neq 0$ is equivalent to the existence of $\epsilon > 0$ such that
    \begin{align}
        \bbeta(t) \in C(\epsilon) \text{ for infinitely many } t\,.
    \end{align}

    Since $C(\epsilon)$ is compact, if we take a subsequence $\bbeta(\tilde t_0), \bbeta(\tilde t_1), \dots$ of $\{\bbeta(t)\}$ such that each $\bbeta(\tilde t_i) \in C(\epsilon)$, this subsequence must have a limit point $\bw$ in $C(\epsilon)$.
    We note that since $\bzero$ and $\bone$ are always equilibrium points, $\bw$ needs be to bounded away from $\bzero$ and $\bone$.
    
There is a sufficiently small $\delta > 0$ such that $U(\bw, \delta) \subset C(\epsilon)$. First, we show that the dynamics $\bbeta(t)$ must leave $U(\bw, \delta)$ in finite time.

Let
\begin{align}
m(\epsilon) = \min_{\bbeta \in C(\epsilon)}  \big(-\langle \boldf(\bbeta) - \bbeta, \nabla V(\bbeta) \rangle\big)
\end{align}
which exists since 
$C(\epsilon)$ is compact. For any $\epsilon > 0$, we know that $m(\epsilon) > 0$ as well, since $\langle \boldf(\bbeta) - \bbeta, \nabla V(\bbeta) \rangle < 0$ for all $\bbeta \not \in B_{\boldf}$ (and is continuous), and $C(\epsilon)$ is compact; thus $-\langle \boldf(\bbeta) - \bbeta, \nabla V(\bbeta) \rangle$ achieves a minimum on $C(\epsilon)$ (rather than just an infimum) which must be positive.

Let $
t^* = 2 c_1 s /m(\epsilon) 
$ (where $c_1$ is from \eqref{eq::err_bound_square_norm})
so that for $t > t^*$
\begin{align}
\label{eq::bound_tstar}
\frac{1}{(t+1)} < \frac{m(\epsilon) }{2c_1s}
\end{align}

Suppose that for $t_0  > t^*$, the dynamics is in $C(\epsilon)$. If the dynamics never leaves $C(\epsilon)$, then
\begin{align}
\lim_{t \to \infty} \sum_{\tau = t_0} ^{t-1} & \frac{1}{\tau + 1} \left\langle  \nabla V(\bbeta(\tau)),  \Delta(\bbeta(\tau)) \right\rangle
\\& \leq \lim_{t \to \infty} \sum_{\tau = t_0} ^{t-1}  \frac{1}{\tau + 1} (- m(\epsilon))
\\ & = - \infty
\end{align}
which contradicts \eqref{eq:finite-drift}. (This fact can also be proven using \cite[Theorem 5.1]{nevel1976stochastic}).
Thus, the trajectory must exit $C(\epsilon)$ in finite time every time it enters.

Next, we use a contradiction to show that the trajectory cannot exit and re-enter $C(\epsilon)$ infinitely often. Since $\bbeta(t)$ leaves $U(\bw, \delta)$ in finite time, it must be that $\bbeta(t)$ enters and exists $U(\bw, \delta)$ infinitely often in order for $\bw$ to be a limit point of a subsequence of $\bbeta (\tilde t_0), \bbeta (\tilde t_1), \dots$.

Now consider $U(\bw, \delta/2)$; since $\bw$ is the limit of subsequence $\bbeta(\tilde t_0), \bbeta(\tilde t_1), \dots$ we know that $\bbeta(t) \in U(\bw, \delta/2)$ for infinitely many $t$. We then define a \emph{passing-through sequence} as a sequence of times $t_k, t_k+1, \dots, t_k+s_k$ such that
    \begin{itemize}
        \item $\bbeta(t_k) \in U(\bw, \delta/2)$;
        \item $\bbeta(t_k+\tau) \in U(\bw, \delta)$ for all $1 \leq \tau \leq s_k$;
        \item $\bbeta(t_k + s_k + 1) \not \in U(\bw, \delta)$.
    \end{itemize}
    
    In other words, this is a sequence of steps of the process that starts when the process is within $\delta/2$ of $\bw$ and ends only when the process goes further than $\delta$ of $\bw$. Since $\bbeta(t)$ visits $U(\bw, \delta/2)$ infinitely many times and leaves $U(\bw,\delta)$ infinitely many times, there must be infinitely many (disjoint) passing-through sequences. We now consider a given passing-through sequence: we know that for any $1 \leq \tau \leq s_k+1$,
    \begin{align}
        \norm{&\bbeta(t_k + \tau) - \bbeta(t_k + \tau - 1)} 
        \\&= \bbnorm{\frac{1}{t_k + \tau} \bpsi(t_k+\tau) - \frac{1}{t_k+\tau}\bbeta(t_k+\tau)}
        \\ &\leq c'(\bw, \epsilon)\left(\frac{1}{t_k + \tau}\right)
    \end{align}
    where $c'(\bw, \epsilon)> 0$ since $\bbeta(t_k + \tau)$ (by virtue of being in $C(\epsilon)$) must be bounded away from $\bzero$ and $\bone$ and $\bpsi(t_k + \tau)$ is a vector of zeros or ones. This equivalently implies that
    \begin{align}
   - \left(\frac{1}{t_k + \tau}\right) \leq - \frac{1}{c'(\bw, \epsilon)} \norm{\bbeta(t_k + \tau) - \bbeta(t_k + \tau - 1)}
    \end{align}
    Meanwhile, we know that since $\bbeta(t_k + \tau-1) \in U(\bw,\delta) \subseteq C(\epsilon)$,
    \begin{align}
        \frac{1}{t_k+\tau}  &\langle \nabla V(\bbeta(t_k+\tau-1)), g(\bbeta(t_k+\tau-1))\rangle 
        \\&< -\frac{1}{t_k + \tau} m(\epsilon) 
        \\ & < - \frac{m(\epsilon)}{c'(\bw, \epsilon)} \norm{\bbeta(t_k + \tau) - \bbeta(t_k + \tau - 1)}
        \\ & = -c \norm{\bbeta(t_k + \tau) - \bbeta(t_k + \tau - 1)}
    \end{align}
    which implies that if $c =   \frac{m(\epsilon)}{c'(\bw, \epsilon)} $, for any $t_k + \tau$ in a passing-through sequence we have
    \begin{align}
    \sum_{\tau=1}^{s_k+1} &\frac{1}{t_k+\tau}   \langle \nabla V(\bbeta(t_k+\tau-1)), g(\bbeta(t_k+\tau-1))\rangle  \\&< -c \sum_{\tau=1}^{s_k+1} \norm{\bbeta(t_k + \tau) - \bbeta(t_k + \tau - 1)} \label{eq:passing-through-bound}
    \end{align} 
    But because a passing-through sequence must start within $U(\bw, \delta/2)$ and end outside $U(\bw, \delta)$,
    \begin{align}
        \norm{\bbeta(t_k+s_k+1) - \bbeta(t_k)} \geq \delta/2
    \end{align}
    which means by the triangle inequality that
    \begin{align}
        \sum_{\tau=1}^{s_k+1} &\norm{\bbeta(t_k + \tau) - \bbeta(t_k + \tau - 1)} 
        \\ &\geq \norm{\bbeta(t_k+s_k+1) - \bbeta(t_k)} \geq \delta/2
    \end{align}
    and hence plugging into equation \eqref{eq:passing-through-bound} yields
    \begin{align}
        \sum_{\tau=1}^{s_k+1} & \frac{1}{t_k+\tau}   \langle \nabla V(\bbeta(t_k+\tau-1)), g(\bbeta(t_k+\tau-1))\rangle  
        \\ & \leq -c \delta/2
    \end{align}
    Since there are infinitely many passing-through sequences, and all the terms of the sum $\sum_{t=1}^\infty \frac{1}{t}  \nabla V(\bbeta(t-1))^\ltop g(\bbeta(t-1))$ are nonpositive we can finally conclude that
    \begin{align}
        \sum_{t=1}^\infty &\frac{1}{t}  \langle \nabla V(\bbeta(t-1)),  g(\bbeta(t-1)) \rangle \\&\leq \sum_{k=1}^\infty \sum_{\tau=1}^{s_k+1} \frac{1}{t_k+\tau}  \nabla V(\bbeta(t_k+\tau-1))^\ltop g(\bbeta(t_k+\tau-1)
        \\ &\leq \sum_{k=1}^\infty (-c\delta/2)
        \\ &= -\infty
    \end{align}
    which contradicts equation \eqref{eq:finite-drift}.

    Thus, we have a contradiction and the proof is complete.

\end{proof}

\fi 

\ifeight 
\fi 

\subsection{\revise{Lyapunov Function for General Networks}}


\newcommand{\hfunc}{f}

\revise{Next, we show general conditions on urn function $\boldf = (\hfunc, \dots, \hfunc)$ so that conditions in \Cref{thm::arthur1986} are met. }

\begin{theorem}\label{thm::lyapunov_monotone}
Let 
\begin{align}
    F(\bbeta, \bgamma) &=  
    \begin{bmatrix}
    \hfunc\left( \mu_1 ,\gamma_{1}\right) \\
    \vdots
    \\
    \hfunc\left( \mu_n,\gamma_{n}\right) 
    \end{bmatrix}
\end{align}
where $\mu_i = \frac{1}{\degree(i)} \sum_{j = 1}^n a_{i,j} \beta_{j}$. Suppose that the network associated with the adjacency matrix $\bA$ is undirected. Then, there exists a \revisered{twice continuously differentiable} Lyapunov function $V$, where $V\geq0$ such that
\begin{align}
    \langle   F(\bbeta,\bgamma) - \bbeta, \nabla V(\bbeta) \rangle \leq 0 \label{eq::inner_product_prop_monoton}
\end{align}
so long as
\begin{itemize}
    \item $\hfunc(\cdot, \gamma)$ is bijective from  $[0,1]$ to $[0,1]$
    \item $\hfunc(\cdot, \gamma)$ is monotonic.
    \item $\hfunc(\cdot, \gamma)$ is \revisered{continuosuly} differentiable
\end{itemize}
Equality in \eqref{eq::inner_product_prop_monoton} holds iff $F(\bbeta,\bgamma)=\bbeta$.
\end{theorem}

\begin{proof}
Because $\hfunc(\cdot, \gamma)$ is bijective from $[0,1]$ to $[0,1]$, there exists an inverse 
$ \hfunc^{-1}(\cdot, \gamma)$.
The function $\hfunc^{-1}(\cdot, \gamma)$ is also (strictly) monotonically increasing.
We use the notation
\begin{align}
    H(\mu, \gamma) = \int_{0}^\mu \hfunc^{-1}(\nu, \gamma) d\nu \,.
\end{align}

Let 
\begin{align}
\label{eq::lyapunov_for_convergence}
V(\bbeta) = \sum_{i = 1}^n \sum_{j = 1}^n a_{i,j} \left(  H\left(\beta_i, {\gamma_i}\right)  -   \frac{1}{2}\beta_i \beta_j \right) + C
\end{align}
(where $C$ is a constant to make $V$ positive). Taking the partial derivatives gives that
\begin{align}
    \frac{\partial V}{\partial \beta_i} &= \degree(i) \hfunc^{-1}\left(\beta_i, {\gamma_i}\right) - \sum_{j =1}^n a_{i,j}\beta_j
    \label{eq::monotone_lyapunov_V_partial_D}
    \\ & = \degree(i) \left( \hfunc^{-1}\left(\beta_i, {\gamma_i}\right) - \frac{1}{ \degree(i)}\sum_{j=1}^n a_{i,j}\beta_j\right)
      \\ & = \degree(i) \left( \hfunc^{-1}\left(\beta_i, {\gamma_i}\right) - \mu_i\right)\,.
\end{align}
(The property that $\bA$ is symmetric is necessary for \eqref{eq::monotone_lyapunov_V_partial_D}.) \revisered{If $\hfunc$ is continuously differentiable then so is $\hfunc^{-1}$ and  $V$ is twice continuously differentiable.}
We can write the $i$th entry in vector $F(\bbeta,\bgamma) - \bbeta$ as
\begin{align}
     (F(\bbeta,\bgamma) - \bbeta)_i =  \hfunc\left( \mu_i,\gamma_{i}\right) - \beta_i \,.
\end{align}
Then
\begin{align}
     &\langle   F(\bbeta,\bgamma) - \bbeta, \nabla V(\bbeta) \rangle 
     \eqlinebreakshort
     \hspace{-2pc} = \sum_{i = 1}^n  \degree(i) \left( \hfunc^{-1}\left(\beta_i, {\gamma_i}\right) - \mu_i\right)\bigg( \hfunc\left( \mu_i,\gamma_{i}\right) - \beta_i\bigg)\,.\label{eq::inner_product_beta_lap}
\end{align}
Suppose $ \hfunc\left( \mu_i,\gamma_{i}\right) > \beta_i$. Then since $\hfunc$ is (strictly) monotone, 
\begin{align}
    \hfunc\left( \mu_i,\gamma_{i}\right) & > \beta_i 
    \\ \iff
    \hfunc^{-1}\left( \hfunc\left( \mu_i,\gamma_{i}\right),{\gamma_i}\right) &> \hfunc^{-1}\left(\beta_i,{\gamma_i}\right)
     \\ \iff
   \mu_i &> \hfunc^{-1}\left(\beta_i,{\gamma_i}\right)\,.
\end{align}
We can conclude that in the case of $ \hfunc\left( \mu_i,\gamma_{i}\right) \neq \beta_i$ the sign of the terms $ \left( \hfunc^{-1}\left(\beta_i, {\gamma_i}\right) - \mu_i\right)$ and $\left( \hfunc\left( \mu_i,\gamma_{i}\right) - \beta_i\right)$ are necessarily different. Hence, their product must be negative.
When $ \hfunc\left( \mu_i,\gamma_{i}\right) = \beta_i$, then the values of both $ \hfunc\left( \mu_i,\gamma_{i}\right) - \beta_i$ and $\hfunc\left( \mu_i,\gamma_{i}\right) - \beta_i$ are zero.

Each term in the sum of \eqref{eq::inner_product_beta_lap} must be nonpositive and thus
\begin{align}
     \langle   F(\bbeta,\bgamma) - \bbeta, \nabla V(\bbeta) \rangle 
     \leq 0 \,.
\end{align}

Equality holds when all terms in the sum of \eqref{eq::inner_product_beta_lap} are zero, which only occurs when  $ \hfunc\left( \mu_i,\gamma_{i}\right) = \beta_i$ for all $i$. This means that \eqref{eq::inner_product_beta_lap} is zero at equilibrium points and only at the equilibrium points. 
\end{proof}

Note that \Cref{thm::lyapunov_monotone} holds for all $\hfunc$ satisfying the given conditions ({not just the specific $f$ defined in \eqref{eq::f_def}), and hence is a general result showing the existence of Lyapunov functions for any dynamics 
satisfying the given conditions on undirected graphs (which do not need to be connected). 

\begin{theorem}\label{thm::must_converge_equilibrium}
  The time-averaged declared opinions $\bbeta(t)$ under the stochastic opinion dynamics governed by \eqref{eq::declared_opinion_prob_2} and \eqref{eq::evolution_stochastic} almost surely \revisered{approaches a set of} equilibrium point of the expected dynamics, that is a fixed point of $F(\cdot,\bgamma)$.
\end{theorem}

\begin{proof}
    \revisered{Follows directly from \Cref{thm::lyapunov_monotone} and \Cref{thm::arthur1986}.
    }

    \revisered{
    Properties of $f(\mu,\gamma)$ defined in \eqref{eq::f_def} include:
\begin{itemize}
    \item $f:[0,1] \to [0,1]$ is bijective in $\mu$ (this can be shown by the fact that $f(f(\mu,1/\gamma),\gamma) = \mu$)
    \item $f$ is monotonic
    \item $f$ is continuously differentiable in $\mu$ \reviseblue{(the derivative is continuous so long as $\gamma > 0$)}.
\end{itemize}
which satisfies the conditions of \Cref{thm::lyapunov_monotone} and those in \Cref{thm::arthur1986}. } 
\end{proof}

\Cref{thm::must_converge_equilibrium} guarantees that almost surely the time-averaged declared opinions for each agent \revisered{$i$ approaches the equilibria of the expected dynamics. Since it is unclear if equilibrium points \revise{(other than $\bzero$ and $\bone$)} are unique or are isolated points, it is possible that each path of the dynamics converges to one of multiple equilibrium points, or the dynamics approaches a set of connected equilibrium points without converging. 
} 

\section{Convergence to Consensus}

\label{sec::converge_to_consensus}

The previous section showed that the opinion dynamics under social pressure almost surely \revisered{approaches a set of equilibria of the expected dynamics, but does not give any details about this equilibria.}
\revise{
Let \emph{boundary equilibrium points} be the points $\bone$ and $\bzero$, while other equilibrium points are \emph{interior equilibrium points}.  Equivalently, an interior equilibrium point is an equilibrium point $\bbeta$ where $0 < \beta_i < 1$ for all $i$. 
}

\revise{
\begin{lemma}\label{lem::all_zeros_or_none}
Suppose that a finite network of agents is connected. Then, if $\bbeta^*$ is an equilibrium point 
such that for some $i$, $\beta_i^* = 0$ (or $\beta_i^* = 1$), then it must be that for all $i$, $\beta_i^* = 0$ (or respectively for all $i$, $\beta_i^* = 1$).
\end{lemma}
\begin{proof}
Suppose that $\beta_i^* = 0$. The only way for this to occur at an equilibrium point is for each of agent $i$'s neighbors $j$ to also have $\beta_j^* = 0$. If any $\beta_j^* > 0$, then $\beta_i^* > 0$ since $\beta_i$ gets a positive contribution from $\beta_j$ in its sum. We continue by inducting on the neighbors of neighbors, and it gives that all agents $j$ in the connected network must have  $\beta_j^* = 0$.
\end{proof} 
}

Since there are multiple equilibrium points (not all necessarily stable) in any opinion dynamics system, in this section we explore conditions under which the system asymptotically converges to a boundary equilibrium point. 
When the system converges to a boundary equilibrium point, 
we say that the agents approach consensus. Consensus occurs when all agents (asymptotically) converge to declaring the same opinion with probability $1$. 

\begin{definition}\label{def::consensus}
\emph{Consensus} is approached if 
\begin{align}
    \bbeta(t) \to \bone \text{ or } \bbeta(t) \to \bzero\quad \text{as}\quad t\to\infty\,.
\end{align}
\end{definition}
In this section, we establish necessary and sufficient conditions for convergence to consensus. In particular,
we show that the probability of approaching consensus is either $0$ or $1$.
 
Recall that $\bbeta$ is \revise{the vector of agents' time-averaged declared opinions.} \Cref{def::consensus} does not imply that any agent will always declare the same opinion, only that her ratio of declared opinions tends to $1$ or $0$. The former statement, in fact, is not true in our opinion dynamics setting.

\begin{lemma}\label{lem::infinitely_many_declare}
Any agent $i$ will almost surely declare infinitely many $0$'s and infinitely many $1$'s.
\end{lemma}

We remark that \Cref{lem::infinitely_many_declare} does not contradict the existence of consensus. Even if agent $i$ declares infinitely many $0$'s and $1$'s, if the ratio of the number of $0$'s declared is less than linear compared to the number of $1$'s declared, then $\beta_i(t) \to 0$.

\Cref{lem::infinitely_many_declare} is fundamentally the same idea as \cite[Proposition 2]{socialPressure2021}, but in our case, we are working with a general network and we emphasize a different aspect of the result. (Recall that the starting condition is always such that $\beta_i(0)$ is not $0$ or $1$, otherwise \Cref{lem::infinitely_many_declare} would not be true. Most of the results in the section crucially depend on this fact.)

%

\begin{proof}[Proof of \Cref{lem::infinitely_many_declare}]

We first create a dummy agent $\fake$ based on agent $i$. Agent $\fake$ makes declarations $\psi_{\fake,t}$ defined by
\begin{align}
\psi_{\fake,t} = \begin{cases}
1 & \text{ w.p } f(\mu_{\fake}(t), \gamma_i)
\\ 0 & \text{ w.p } 1 - f(\mu_{\fake}(t), \gamma_i)
\end{cases}
\end{align}
where unlike agent $i$, we fix that
$
\mu_{\fake}(t) = \frac{m_i^{1}}{\degree(i) t} 
$
for each time step
(recall that $m_i^1$ is an initialization for agent $i$, and our model assumes $m_i^1 > 0$). 
Let $m_{\fake} \eqdef  {m_i^{1}}/{\degree(i) }$.
We have that
\begin{align}
\sum_{t = 2}^{\infty} \bbP[\psi_{\fake,t} = 1] & = \sum_{t = 2}^{\infty} f(m_{\fake}/ t, \gamma_i)
= \sum_{t = 2}^{\infty} \frac{\gamma_i m_{\fake} }{t + (\gamma_i - 1) m_{\fake}}\,.
\end{align}

If $(\gamma_i - 1)$ is negative, then let $t_0$ be such $t_0 + (\gamma_i - 1) m_{\fake} \geq 1$, which gives
\begin{align}
\sum_{t = 2}^{\infty} \frac{\gamma_i m_{\fake} }{t + (\gamma_i - 1) m_{\fake}} \geq \sum_{t = t_0} ^{\infty} \frac{\gamma_i m_{\fake} }{t + (\gamma_i - 1) m_{\fake}} \geq \sum_{t' = 1}^{\infty} \frac{c}{t'} = \infty.
\end{align}
If $(\gamma_i - 1)$ is positive, then
\begin{align}
    \sum_{t = 2}^{\infty} \frac{\gamma_i m_{\fake} }{t + (\gamma_i - 1) m_{\fake}} \geq 
    \sum_{t = 1}^{\infty} \frac{\gamma_i m_{\fake}}{t} = \infty.
\end{align}

Because $\sum_{t = 2}^{\infty} \bbP[\psi_{\fake,t} = 1] = \infty$ and each declaration is independent for agent $\fake$, using the (second) Borel-Cantelli lemma gives that $\psi_{\fake,t}$ is $1$ infinitely often almost surely.

Next, we couple agent $i$'s declaration with agent $\fake$'s. Since 
\ifeight
$\mu_i(t) \geq \frac{m_{\fake}}{t} = \mu_{\fake}(t)$
\fi
we can create a joint distribution where
$\psi_{i, t} = 1$ if $\psi_{\fake,t} = 1$
and $\psi_{i, t} = 1$ with probability 
\ifeight
$\frac{f(\mu_{\fake}(t), \gamma_i) - f(\mu_i(t), \gamma_i)}{1 - f(\mu_i(t), \gamma_i)}$
\fi
if $\psi_{\fake,t} = 0$.
The marginal distributions on this coupling shows that $\psi_{i,t}$ is $1$ more often than $\psi_{\fake,i}$ is $1$, thus
 $\psi_{i,t}$ must be $1$ infinitely often almost surely. 
 By symmetry, the same result holds for declaring infinitely many $0$'s.
\end{proof}


\revise{Whether or not $\bbeta(t)$ converges to a boundary equilibrium point} is closely related to the Jacobian matrix of $F(\cdot,\bgamma)$. 
\ifeight
To calculate the Jacobian, we have
\fi
\begin{align}
    \frac{\partial}{\partial \mu_i} F_i(\bbeta,\bgamma) = 
\frac{\gamma_i}{(1+(\gamma_i - 1)\mu_i)^2}
\end{align}
and thus
\begin{align}
\frac{\partial}{\partial\bbeta}F(\bbeta,\bgamma) &= \frac{\partial}{\partial\bmu}F(\bbeta,\bgamma) 
\frac{d\bmu}{d \bbeta}
\\&
= \diag\left(\begin{bmatrix}
     \frac{\gamma_1}{(1+(\gamma_1 - 1)\mu_1)^2}
     \\
     \vdots
     \\
      \frac{\gamma_n}{(1+(\gamma_n - 1)\mu_n)^2}
    \end{bmatrix}\right) \bW.
\end{align}

We define for each vector $\bx$ where $x_i \in [0, 1]$, 
\begin{align}
    \label{eq::dermat_def}
    \dermat{\bx} \eqdef \diag\left(\begin{bmatrix}
     \frac{\gamma_1}{(1+(\gamma_1 - 1)\frac{1}{\degree(1)} \sum_{j } a_{1,j} x_1)^2}
     \\
     \vdots
     \\
      \frac{\gamma_n}{(1+(\gamma_n - 1)\frac{1}{\degree(n)} \sum_{j} a_{n,j} x_n)^2}
    \end{bmatrix}\right) \bW \,.
\end{align}
Importantly, $\dermat{\bx}$ has all real eigenvalues. 

The Jacobian at boundary equilibrium points $\bzero$ and $\bone$ are
\begin{align}
\dermat{\bone} = \frac{\partial}{\partial\bbeta}F(\bbeta,\bgamma)|_{\bbeta = \bone} &=  \bGamma^{-1}\bW\label{eq::dermat_1}
\\ \dermat{\bzero} = \frac{\partial}{\partial\bbeta}F(\bbeta,\bgamma)|_{\bbeta = \bzero} &= \bGamma\bW ,
\end{align}
where $\bGamma$ is defined in \eqref{eq::diag_gamma_matrix}. 

We will prove that determining properties of $\dermat{\bzero}$ and  $\dermat{\bone}$ suffices to determine whether $\bbeta(t)$ approaches a boundary equilibrium point or not. In particular, if the eigenvalues of the Jacobian evaluated at an equilibrium point are all less than $1$, then the equilibrium point is a stable equilibrium point. 
To do this, in the next step we establish that for boundary equilibrium points $\bx$, if the largest eigenvalue of $\dermat{\bx}$ is larger than $1$, then the full stochastic process cannot converge to $\bx$. Several intermediate results need to be shown in order to prove this. The main tool is given by 
\cite[Theorem 3]{renlund2010generalized} 
and this is also a key result used by \cite{socialPressure2021}. We state it with notation adjusted\footnote{Two of the conditions originally stated in \cite[Theorem 3]{renlund2010generalized} are  combined to make one condition in our statement.} for our use.

\begin{proposition}[Adapted from \cite{renlund2010generalized}]\label{prop::conditions_zero_prob}
Suppose there is a stochastic process $X(t)$ where $0 \leq X(t) \leq 1$, a filtration $\cF_t$, and a $\varepsilon > 0$ with the following properties:
\begin{enumerate}[label=(\alph*)]
\item \label{item::conditions_Xgreaterep} If $X(t) \leq \varepsilon$,
$
\bbE[X(t+1) - X(t)|\cF_t] \geq 0
$
\item \label{item::conditions_var}$\var[X(t+1) - X(t)|\cF_t] \leq c_1 \frac{X(t)}{(t+1)^2}$ 
\item \label{item::conditions_inf}$\lim_{t \to \infty} t X(t) = \infty$ with probability $1$ \,.
\end{enumerate}
Then
$\bbP\left[\lim_{t \to \infty} X(t) = 0\right] = 0$.
\end{proposition}
%
%
%
\ifshowproof 
\fi 

\ifshowproof 
The proof below is included for completeness.

\begin{proof}

Let $t_0$ be some time where 
$X(t_0) \leq \varepsilon$. If this does not exist, then $X(t)$ cannot converge to $0$.
Let $a \in (0, 1)$ and
define the following stopping times in terms of $t_0$:
\begin{align}
\tau_1 &= \inf\{t > t_0: \, X(t) \geq 2 X(t_0) \wedge \varepsilon  \}
\\\hat \tau_1 &= \inf\{t > t_0: \,  X(t) \leq a X(t_0) \}\,.
\end{align}

Let event $A = \{ \hat \tau_1 < \tau_1 \}$ and $B =\{ \tau_1 < \hat \tau_1 \} $. 
\revise{The event $A$ represents the case where the trajectory $X(t)$ gets to be twice the value of $X(t_0)$ (or gets larger than $\epsilon$) before it gets factor $a$ closer to $0$.} \revise{Let $\tau = \tau_1 \wedge \hat \tau_1$.}
Define
\begin{align}
Y(t) &= X(t+1) - X(t)
\\Z(t) &= \bbE[Y(t)| \cF_t] - Y(t)
\\W(t) &= \sum_{s = t_0 }^{t-1} Z(s)\,.
\end{align}
\revise{and thus $W(t)$ is a martingale process.}

On event $A$, we have that
\begin{align}
W(\hat\tau_1) 
&= \sum_{s = t_0}^{\hat\tau_1 - 1}\bbE[Y(s) | \cF_t]  - \sum_{s = t_0 }^{\hat\tau_1 - 1}(X(s+1) - X(s))
\\& \geq 0 - X(\hat \tau_1) + X(t_0)
\\ & \geq ( 1 - a) X(t_0) 
\end{align}
where we used property \ref{item::conditions_Xgreaterep}.
This gives that
\begin{align}
\bbE[W(\tau) | A] \geq (1-a) X(t_0)\,.
\end{align}

Next, using independence and property \ref{item::conditions_var}
\begin{align}
\bbE[W(t)^2] &= \sum_{s = t_0}^{t-1} \bbE[Z(s)^2]
\\ & = \sum_{s = t_0}^{t-1} \var[Y(s)]
\\ & \leq \sum_{s = t_0}^{t-1} c_1 \frac{X(s)}{(s+1)^2}\,.
\end{align}
Since $X(s) \leq 2 X(t_0) \wedge \varepsilon \leq 2 X(t_0)$, this gives that
\begin{align}
\bbE[W(\tau)^2] &\leq c_1 2 X(t_0) \sum_{s = t_0}^{\tau-1} \frac{1}{(s+1)^2}
\\&\leq c_1 2 X(t_0) \sum_{s = t_0}^{\infty} \frac{1}{(s+1)^2}
\\&\leq c_1 2 X(t_0) \frac{1}{t_0}
\,.
\end{align}

Next we show the inequality derived in \cite[Lemma 6]{renlund2010generalized}.
Since $W$ is a martingale
\begin{align}
0 &= \bbE[W(\tau)] = \bbE[W(\tau)| A] \bbP[A] + \bbE[W(\tau) | B] \bbP[B]
\\& \implies ( \bbE[W(\tau) | B])^2 =  (\bbE[W(\tau)| A])^2 \frac{(\bbP[A])^2}{(\bbP[B])^2}
\end{align}
Then
\begin{align}
\bbE[W(\tau)^2] 
& = \bbE[W(\tau)^2| A] \bbP[A] + \bbE[W(\tau)^2 | B] \bbP[B]
\\&\geq (\bbE[W(\tau)| A])^2 \bbP[A] + (\bbE[W(\tau) | B])^2 \bbP[B]
\\& = (\bbE[W(\tau)| A])^2 \bbP[A] + (\bbE[W(\tau)| A])^2 \frac{(\bbP[A])^2}{(\bbP[B])}
\\& = (\bbE[W(\tau)| A])^2 \frac{\bbP[A]}{\bbP[B]} \left(\bbP[B] + \bbP[A] \right)
\\& = (\bbE[W(\tau)| A])^2 \frac{\bbP[A]}{\bbP[B]} 
\end{align}

Using the above inequality
, we have that
\begin{align}
\frac{\bbP[B]}{\bbP[A]} &\geq \frac{(\bbE[W(\tau)| A])^2}{\bbE[W(\tau)^2]}
\\ &\geq \frac{((1-a) X(t_0))^2}{2 c_1 X(t_0) \frac{1}{t_0}}
\\ & \geq c(a) t_0  X(t_0)
\end{align}
where 
\begin{align}
c(a) = \frac{(1-a)^2}{2c_1}\,.
\end{align}

\revise{Since $B = A^{c}$, then $\bbP[A] = 1 - \bbP[B]$, and 
\begin{align}
\frac{\bbP[B]}{1 - \bbP[B]} > \kappa \implies \bbP[B] \geq 1 - \frac{1}{1+\kappa}
\end{align}
which gives that}
\begin{align}\label{eq::prob_B_stoppingtime1}
\bbP[B] \geq 1 - \frac{1}{1 + c(a) t_0 X(t_0)}
\,.
\end{align}

We define stopping times recursively as
\begin{align}
\tau_{k + 1} &= \inf \{ t \geq \tau_k: \, X(t) \geq 2 X(\tau_k ) \wedge \varepsilon \}
\\ \hat \tau_{k+1} &= \inf \{ t \geq \tau_k: \, X(t) \leq a X(t_0) = a_k X(\tau_k)  \}
\end{align}
where $a_k$ is a stochastic value such that $a X(t_0) = a_k X(\tau_k) $ such that $a_k \leq a$. This implies that $c(a_k) > c(a)$.
Define the events $B_k = \{ \tau_k  \leq \hat \tau_{k}\}$ and $T_k = \tau_k \wedge \hat \tau_k$.

Using \eqref{eq::prob_B_stoppingtime1},
if $X(\tau_k) \geq 2^k X(t_0) $ (and $X(\tau_k) < \varepsilon$),
\begin{align}
\bbP[B_{k+1}| B_{k}, \cF_{T_k}] 
&\geq 1 - \frac{1}{1 + c(a_k) \tau_k X(\tau_k)}
\\ 
& \geq 1 - \frac{1}{1+c(a) 2^k t_0 X(t_0)}
\,.
\end{align}

If $X(\tau_k) \geq \varepsilon$, then $\bbP[B_{k+1}| B_{k}, \cF_{T_k}] = 1$. 

Then
\begin{align}
\bbP\left[\sup_{t > t_0}{X_t} > \varepsilon\right] &\geq \bbP\left[ \bigcap_{k}^\infty B_k\right]
\\& \geq \prod_{k = 1}^{\infty} 1 - \frac{1}{1 + c(a) 2^k t_0 X(t_0)} 
\\ &\geq 1 - \sum_{k = 1}^\infty \frac{1}{1+c(a) 2^k t_0 X(t_0)} 
\\ &\geq 1 - \sum_{k = 1}^\infty \frac{1}{c(a) 2^k t_0 X(t_0)} 
\\ &= 1 - \frac{2}{c(a) t_0 X(t_0) }\,.
\end{align}

To show that the above approaches $1$ as $t_0 \to \infty$, we use property \ref{item::conditions_inf}.
%
This shows that with probability $1$, if $X(t) < \varepsilon$ at any time $t_0$, there is another time $t > t_0$ where $X(t) \geq \varepsilon$. Thus, with probability $1$, $X(t)$ cannot converge to $0$.
\end{proof}
\fi 
%
%

To show our desired result, we need to find a process $X(t)$ based on the value of $\bbeta$ which fits the conditions of \Cref{prop::conditions_zero_prob}. This is done in the next lemma.

\begin{lemma}\label{lem::meet_conditions}
Suppose that $\dermat{\bx}$ for $\bx \in \{\bzero, \bone\}$ has an eigenvalue larger than $1$. Then there exists a $\varepsilon > 0$ and a function $V: [0,1]^n \to [0,1]$ such that $V(\bbeta)$ is a stochastic process satisfying all the conditions of \Cref{prop::conditions_zero_prob}. 
\end{lemma}

\begin{proof}

For this proof, we let $\bx = \bzero$ (the proof is symmetric when $\bx = \bone$).  Let $\lambda = \lambda_{\max}(\dermat{\bzero})$. By assumption, $\dermat{\bzero}$ has an eigenvalue greater than $1$, so $\lambda > 1$. Let $\bv$ be the associated (left) eigenvector of eigenvalue $\lambda$. Since $\dermat{\bzero}$ is irreducible (since $\bW$ is irreducible) and a nonnegative matrix, the 
Perron-Frobenius theorem \cite{Berman1994} implies that 
$\bv$ is a positive eigenvector. Scale $\bv$ so that $\bv^\ltop \bone = 1$. 
Let
$
V(\bbeta) = \bv ^{\ltop} \bbeta\,.
$

We determine a small enough value for $\varepsilon$ in the next part.
We first derive a lower bound on $f$:
\begin{align}
f(\mu_i, \gamma_i) &\geq \gamma_i \mu_i - \frac{1}{2}c_1 \mu_i^2
\label{eq::lower_bound_f}
\end{align}
where $ c_1 = 2 \gamma_i (\gamma_i - 1)$. We can show this result by showing that
\ifeight
$\frac{\partial^2}{\partial \mu_i^2} f(\mu_i, \gamma_i) \geq -c_1$.
\fi
To get this, we use
\begin{align}
\frac{\partial^2}{\partial \mu_i^2} f(\mu_i, \gamma_i) = \frac{-2 \gamma_i (\gamma_i - 1)}{(1 + (\gamma_i - 1) \mu_i )^3}\,.
\label{eq::second_deriv_f}
\end{align}
If $\gamma_i > 1$, the numerator of \eqref{eq::second_deriv_f} is negative. To get a lower bound, we want to minimize the denominator. This occurs when $\mu_i = 0$ and the denominator of \eqref{eq::second_deriv_f} is $1$.
Likewise, if $\gamma_i < 1$, the numerator of \eqref{eq::second_deriv_f} is positive, so we want to maximize the denominator. This happens again at $\mu_i = 0$ resulting in a denominator of $1$.
(If $\gamma_i = 0$, then $f$ is linear and $c_1 = 0$.)

As a result of \eqref{eq::lower_bound_f}, there is a $\delta > 0$, such that 
\ifeight
$ f(\mu_i, \gamma_i) > \frac{\gamma_i \mu_i}{\lambda} $
\fi
for all $\mu_i < \delta$ and all agents $i$.
Recall that
\begin{align}
    \bbeta(t+1) - \bbeta(t) &= \frac{1}{t+1} \left(F(\bbeta(t), \bgamma) - \bbeta(t)\right) 
    + \frac{1}{t+1} \boldeta(t+1)
\end{align}
where $\boldeta(t+1)$ is a vector with $i$th component
\begin{align}
    \eta_i(t+1) = \begin{cases}
    1 - f(\mu_i(t), \gamma_i) & \text{w.p. }  f(\mu_i(t), \gamma_i) \\
    - f(\mu_i(t), \gamma_i) & \text{w.p. }  1 - f(\mu_i(t), \gamma_i)
    \end{cases}
\end{align}
and
\ifeight
$ \bbE[\bbeta(t+1) - \bbeta(t)] = \frac{1}{t+1} \left(F(\bbeta(t), \bgamma) - \bbeta(t)\right)$.
\fi

Define
$
Y(t) = V(\bbeta(t+1)) - V(\bbeta(t))\,.
$
We can pick $\varepsilon$ small enough so that $V(\bbeta) < \varepsilon$ implies that $\beta_i < \delta$ for all agents $i$, which then implies that $\mu_i < \delta$ for all agents $i$.

Note that
\begin{align}
&\bbE [Y(t)| \cF_t]  
= \bv^\ltop  \frac{1}{t+1} \left(F(\bbeta(t), \bgamma) - \bbeta(t)\right) 
\\&\geq \frac{1}{t+1}\bv^\ltop \left(\frac{1}{\lambda}\dermat{\bzero} -\bI\right) \bbeta(t)
= \frac{1}{t+1}\left(\frac{\lambda}{\lambda} - 1 \right) \bv^\ltop \bbeta(t)
 = 0
\end{align}
This shows property \ref{item::conditions_Xgreaterep} as stated in \Cref{prop::conditions_zero_prob}.

We compute that for all  $\bbeta$,
\ifeight 
\begin{align}
\var[\eta_i(t+1)] 
& \leq f(\mu_i(t), \gamma_i)
 \leq \max\{\gamma_i, 1/\gamma_i \} \mu_i(t)
\end{align}
\fi 
Let $c_0 = \max_{i} \left\{\frac{v_i}{\gamma_i} \max\{\gamma_i, 1/\gamma_i \}\right\}$.
\ifeight 
Then
\begin{align}
&\var[Y(t)|\cF_{t}] 
= \frac{1}{(t+1)^2}\sum_{i = 1}^n v_i^2\var\left[ \eta_i(t+1) \right]
\\ & \leq   \frac{1}{(t+1)^2}\sum_{i = 1}^n v_i^2 \max\{\gamma_i, 1/\gamma_i \} \mu_i(t)
\\ & \leq   \frac{1}{(t+1)^2}\sum_{i = 1}^n c_0 v_i \gamma_i  \frac{1}{\degree(i)}\sum_{j} {a_{i,j}} \beta_j(t)
\\ & 
=   \frac{c_0 \lambda }{(t+1)^2}V(\bbeta(t))
\,.
\end{align}
\fi  
This shows property \ref{item::conditions_var} of \Cref{prop::conditions_zero_prob}.
%
%
%
We write
\begin{align}
t_0 \beta_i(t_0) &=t_0 \frac{b_i^0 + \sum_{s = 1}^{t_0} (1 - \psi_{i, s})}{t_0}\,.
\end{align}
By \Cref{lem::infinitely_many_declare}, we know $\sum_{s = 1}^{t_0} (1 - \psi_{i, s}) \to \infty$ for each $i$, so 
\begin{align}
t_0 V(\bbeta(t_0)) 
 = \sum_{i = 1}^n v_i \left(b_i^0 + \sum_{s = 1}^{t_0} (1 - \psi_{i,s}) \right)
 \to \infty
\end{align}
as $t_0 \to \infty$, which shows property \ref{item::conditions_inf} of \Cref{prop::conditions_zero_prob}.
\end{proof}

\begin{lemma}\label{lem::only_one_eig_smaller}
Only one of $\dermat{\bzero}$ and $\dermat{\bone}$ can have all eigenvalues less than or equal to $1$. 
\end{lemma}

\begin{proof}
Let $\bM_{\bzero} = \bD^{-1/2}\bGamma^{1/2}  \bA \bGamma^{1/2} \bD^{-1/2}$ and $\bM_{\bone}   = \bD^{-1/2}\bGamma^{-1/2}  \bA \bGamma^{-1/2} \bD^{-1/2}$. Define $m = \sum_{i} \degree(i)$ and let
\ifeight
$\bx = \left[\sqrt{\frac{\degree(1)}{m}}, \sqrt{\frac{\degree(2)}{m}}, \dots, \sqrt{\frac{\degree(n)}{m}}  \right]^{\ltop}$
\fi
so that $\lVert \bx \rVert_2 = 1$.
\begin{align}
\lambda_{\max}(\dermat{\bzero}) &=
\lambda_{\max}(\bGamma \bW) = 
\lambda_{\max}(\bGamma \bD^{-1} \bA) 
\\&= \lambda_{\max}(\bD^{-1/2}\bGamma^{1/2}  \bA \bGamma^{1/2} \bD^{-1/2})
= \lambda_{\max} (\bM_{\bzero})
\end{align}

If $\lambda_{\max} (\bM_{\bzero}) \leq 1$, then 
\ifeight
$\frac{1}{m} \sum_{i,j} a_{i,j} \sqrt{\gamma_i \gamma_j} =
\bx^{\ltop} \bM_{\bzero} \bx  \leq 1$.
\fi
Similarly if  $\lambda_{\max}(\dermat{\bone}) =\lambda_{\max}(\bM_{\bone}) \leq 1$, then
\ifeight
$\frac{1}{m} \sum_{i,j} a_{i,j} \frac{1}{\sqrt{\gamma_i \gamma_j}} =
\bx^{\ltop} \bM_{\bone} \bx  \leq 1$.
\fi
Using Cauchy-Schwarz gives that 
\begin{align}
\Bigg(\frac{1}{m} &\sum_{i,j} a_{i,j} \sqrt{\gamma_i \gamma_j} \Bigg) \Bigg(\frac{1}{m} \sum_{i,j} a_{i,j} \frac{1}{\sqrt{\gamma_i \gamma_j}}\Bigg) \\&\geq \frac{1}{m^2}\left( \sum_{i,j} a_{i,j} \sqrt{\frac{\sqrt{\gamma_i \gamma_j} }{\sqrt{\gamma_i \gamma_j} }}\right)^2 = 1
\label{eq::cs_for_eigen_less_one}
\,.
\end{align}
For equality to hold there needs to be some constant $c$ where 
$
\sqrt{\gamma_i \gamma_j} = c / {\sqrt{\gamma_i \gamma_j}}
$
for all $i$ and $j$. Since the graph is not bipartite, this only holds if $\gamma_1^2 = \dots = \gamma_n^2  = c$ for all $i$ and $j$. In the case where $c > 1$ then $\lambda_{\max}(\bGamma \bW) > 1$; if $c < 1$ then $\lambda_{\max}(\bGamma^{-1} \bW) > 1$; if $c = 1$, then  $\gamma_1 = \dots = \gamma_n  = 1$ which is not allowed by our assumptions. Thus, equality in \eqref{eq::cs_for_eigen_less_one} cannot hold.  
Therefore, the statements $\lambda_{\max} (\bM_{\bzero}) \leq 1$ and $\lambda_{\max} (\bM_{\bone}) \leq 1$ cannot both be true.
\end{proof}

Note that the assumptions $\bGamma \neq \bI$ and that the graph is not bipartite are critical for this lemma. 
If the graph is a path with no self-loops (which is bipartite) where adjacent nodes alternate bias parameters $\gamma$ and $1/\gamma$, then the largest eigenvalues of $\dermat{\bone}$ and $\dermat{\bzero}$ can both be $1$.

\begin{theorem}\label{thm::eig_does_not_converge}
    Let $\bx$ be a boundary equilibrium point  (either $\bzero$ or $\bone$). If $\lambda_{\max}(\dermat{\bx}) > 1$, then
    $   
        \bbP[\bbeta(t) \to \bx] = 0.
    $

    Conversely, if $\lambda_{\max}(\dermat{\bx}) \leq 1$, then
    $   
        \bbP[\bbeta(t) \to \bx] = 1.
    $
\end{theorem}

\begin{proof}
For the first statement, by combining \Cref{lem::meet_conditions} and \Cref{prop::conditions_zero_prob}, we can show that if $\dermat{\bzero}$ has an eigenvalue greater than $1$, the result holds. By symmetry, if $\dermat{\bone}$ has an eigenvalue greater than $1$, the result also holds.

For the second statement, our main method is to show that no interior equilibrium point exists if one of $\dermat{\bzero}$ and $\dermat{\bone}$ has all eigenvalues less than or equal to $1$. 
We assume that $\dermat{\bone}$ has all eigenvalues less than or equal to $1$. (By symmetry, the same proof can be used for the case that $\dermat{\bzero}$ has all eigenvalues less than $1$.)
Let $\lambda = \lambda_{\max}(\dermat{\bone})$ and let $\bv$ be the corresponding eigenvector. Let $v_i$ be the $i$th element of $\bv$.  Scale $\bv$ so that $\bv^T \bone = 1$. We will use the fact $v_i > 0$ (shown by Perron-Frobenius) and $\bv ^T \dermat{\bone}=\lambda\bv ^T  $.

Observe that (as $\frac{1}{\bbE[X]} \leq \bbE[1/X] $ by Jensen's inequality)
\begin{align}
    \sum_{i = 1}^n & v_i\left(\frac{1}{f(\mu_i,\gamma_i)}-1\right) 
    = \sum_{i = 1}^n  \frac{v_i}{\gamma_i}\left(\frac{1}{\mu_i}-1\right)\\
    &\leq
    \sum_{i = 1}^n  \frac{v_i}{\gamma_i \degree(i)}\sum_{j = 1}^n a_{i,j} \left(\frac{1}{\beta_j}-1\right)\label{eq::boundary_urn_proof_convex}
    \,.
\end{align}

    Then
\begin{align}
    \sum_{i = 1}^n & v_i\left(\frac{1}{f(\mu_i,\gamma_i)}-1\right)  
    \leq
    \sum_{j = 1}^n \left(\frac{1}{\beta_j}-1\right) \sum_{i = 1}^n \frac{v_i}{\gamma_i \degree(i)} a_{i,j}\\
    &=\lambda
    \sum_{j = 1}^n v_j\left(\frac{1}{\beta_j}-1\right) \leq \sum_{j = 1}^n v_j\left(\frac{1}{\beta_j}-1\right)
\end{align}

\begin{align}
    \implies \sum_{i = 1}^n \frac{v_i}{f(\mu_i,\gamma_i)} &\leq 
    \sum_{i = 1}^n \frac{v_i}{\beta_i}.
\label{eq::inverse_ineq}
\end{align}
Interior equilibrium points must have that $\beta_i = f(\mu_i,\gamma_i)$ for all $i$. Inequality \eqref{eq::inverse_ineq} is a strict inequality when $\lambda < 1$, in which case there must not exist any interior equilibrium points $\bbeta$.
When $\lambda = 1$, \eqref{eq::inverse_ineq} can only be an equality if \eqref{eq::boundary_urn_proof_convex} is an equality. Since $1/x$ is a strictly convex function, equality in \eqref{eq::boundary_urn_proof_convex} only holds if all $\beta_j$'s are equal for all $j$ which is a neighbor of $i$. Since the graph is not bipartite and connected, this implies that all $\beta_j$ are the same for each $j$. (We can see this since the non-bipartite property implies that there is a path with an even number of nodes connecting any node $i$ to node $j$. The nodes at odd positions in the path will force the pair of two adjacent even position nodes to be the same.)
However, the only way $\bbeta$ can be an equilibrium point with this condition that $\beta_j$'s are all equal 
is if $\bbeta = \bone$ or $\bbeta = \bzero$, or if $\gamma_i = 1$ for all $i$. 
Thus, if $\lambda < 1$ or $\lambda = 1$, the only equilibrium points are $\bone$ and $\bzero$.

Using \Cref{thm::must_converge_equilibrium}, with probability $1$, $\bbeta(t)$ must converge to one of the two boundary equilibrium point. If $\bx$ is such that $\dermat{\bx}$ has all eigenvalues less than or equal to one, then by 
\Cref{lem::only_one_eig_smaller}, the other boundary equilibrium point must have an eigenvalue larger than $1$. Using the first statement of this theorem, $\bbeta(t)$ almost surely cannot converge to this other boundary equilibrium point, and thus  $\bbeta(t)$ converges to $\bx$ with probability $1$.
\end{proof}

\Cref{thm::eig_does_not_converge}  gives the answer to when the opinion dynamics converges to consensus or not. The only property that needs to be checked are the eigenvalues of $\dermat{\bzero}$ and $\dermat{\bone}$. 
We remark that if all agents' bias parameters are greater than $1$ ($\bGamma^{-1}$ has all values less than $1$), we immediately get that $\dermat{\bone}$ has largest eigenvalue less than $1$. Thus, as expected, all agents converge to declaring opinion $1$ with probability $1$. 
In \cite{socialPressure2021}, the threshold for approaching consensus is computed when the network is the complete graph, where proportion $\Phi$ of the agents have bias parameter $\gamma > 1$ and proportion $1 - \Phi$ of the agents have bias parameter $1/\gamma$. Consensus occurs if and only if 
\ifeight
$\gamma \leq \max\left\{\frac{1 - \Phi}{\Phi}, \frac{\Phi}{1-\Phi} \right\}$.
\fi
This is consistent with \Cref{thm::eig_does_not_converge}. 

\ifeight 

\fi 

\ifeight 

\fi 

\section{Conclusion}

\label{sec::conclusion}

In this work, we studied the interacting P\'{o}lya urn model of opinion dynamics under social pressure. 
We expanded upon \cite{socialPressure2021} by showing results for arbitrary networks and general bias parameters.
To show that the probability of declared opinions \revisered{almost surely approach equilibria, we used an appropriate Lyapunov function and applied stochastic approximation}.
We also gave easily-computable necessary and sufficient conditions, 
for when the dynamics approach consensus. 
Our results provide insight as to how and when social pressure can force conformity of (expressed) opinions even against the true beliefs of some individuals.

The convergence and consensus results developed in this work have potential applications beyond this opinion dynamics model. These results may also apply to other social dynamics similar to the interacting P\'{o}lya urn model with non-linear interaction functions. A possible direction for further work is to find what consequences our techniques have for other models. Finding the interior equilibrium points, \revise{and determining if they are isolated or unique,} for arbitrary networks is also an area for future work.

\bibliographystyle{resources/IEEEbib}
\bibliography{ref_social}


\ifeight 
\fi 

\ifsixteen

\else
\subsection{Proofs and Discussion for the Simplified Community Network}

\subsubsection{Proof of \Cref{prop::2agent_interior_eq_point}}

\label{sec::proof_2agent_interior_eq_point}

\begin{proof}
To find the equilibrium points, we solve a set of equations with two variables. We work with $\mu_a$ and $\mu_b$ (see \eqref{eq::mu_vt}).
\begin{align}
    \mu_a &= p_1 f(\mu_a, \gamma) + (1-p_1)  f(\mu_b, 1/\gamma) \\
    &= p_1 \frac{\gamma \mu_a }{1 + (\gamma - 1) \mu_a} + (1-p_1) \frac{\mu_b}{\gamma - (\gamma - 1)\mu_b} \label{eq::comm_mua_evolve}\\
    \mu_b &= (1-p_2) f(\mu_a,\gamma) + p_2  f(\mu_b, 1/\gamma) \\
    &= (1-p_2) \frac{\gamma \mu_a }{1 + (\gamma - 1) \mu_a} + p_2 \frac{\mu_b}{\gamma - (\gamma - 1)\mu_b}\label{eq::comm_mub_evolve}\,.
\end{align}

Substituting $u_a$ and $u_b$ for one another to gives
\begin{align}
    \mu_b = \frac{1 - p_1 - p_2}{1 - p_1} \frac{(\gamma - 1) \mu_a (1-\mu_a)}{1+ (\gamma - 1) \mu_a} + \mu_a
\end{align}
and
\begin{align}
    \label{eq::2agent_eq_quadratic}
   &\mu_a(1-\mu_a)\bigg((\gamma - 1)^2 \mu_a^2 p_2 + \mu_a\big(2(\gamma - 1)p_2 -(\gamma^2 - 1) p_1p_2\big) 
   \eqlinebreakshort
   -\gamma (1-p_1) p_1 + p_2 - p_1 p_2\bigg)  = 0\,.
\end{align}
Equation \eqref{eq::2agent_eq_quadratic} has the solutions
$\mu_a = 0, \mu_a = 1$ and $\nu_{-}$ and $\nu_{+}$ where
\begin{align}
    \Delta & = \sqrt{4 \gamma (1-p_1 - p_2) + (\gamma +1)^2 p_1 p_2}
    \\ \nu_{-} & = \frac{(\gamma + 1) p_1 - 2 - \sqrt{\frac{p_1}{p_2}}\Delta}{2 (\gamma - 1)}
    \\ \nu_{+} & = \frac{(\gamma + 1) p_1 - 2 + \sqrt{\frac{p_1}{p_2}}\Delta}{2 (\gamma - 1)}
    \,.
\end{align}

The values of $\nu_{-}$ and $\nu_{+}$ need to be between $0$ and $1$ in order to be equilibrium points.

We first check these conditions for $\nu_{-}$.
To show that $\nu_{-} \geq 0$, we need to check that
\begin{align}
    (\gamma + 1) p_1 - 2 \geq \sqrt{\frac{p_1}{p_2}} \Delta \label{eq::solve_muastarminus}
\end{align}
(It is important to note that $\gamma >1$. The above and other inequalities can be different if $\gamma < 1$.)

We can solve for conditions on $\gamma$ by squaring the equation, but this requires that we meet the requirement that
\begin{align}
    \label{eq::2agent_b_pos_condition}
    (\gamma + 1) p_1 - 2 > 0 \implies \gamma \geq \frac{2}{p_1} - 1\,.
\end{align}
Squaring \eqref{eq::solve_muastarminus} gives an expression which is equivalent to 
\begin{align}
    \gamma \leq \frac{p_2}{p_1}.
\end{align}
However, this means we need to meet the condition that 
\begin{align}
\gamma \leq \frac{p_2}{p_1} < \frac{2}{p_1} - 1 \leq \gamma
\end{align}
which is not possible (since $p_1$ and $p_2$ are in $[0,1]$). Thus, $\nu_{-}$ is always less than $0$ and we do not need to consider it as a possible equilibrium point. 

Next, we check the conditions for when $\nu_{+}$ is between $0$ and $1$. We start by checking when $\nu_{+} \geq 0$, which is when
\begin{align}
    (\gamma + 1) p_1 - 2 \geq - \sqrt{\frac{p_1}{p_2}} \Delta \label{eq::solve_muastarminus_pos}
    \,.
\end{align}

This is true if either \eqref{eq::2agent_b_pos_condition} holds or, as we compute by squaring both sides of \eqref{eq::solve_muastarminus_pos}, if
\begin{align}
    \label{eq::2agent_condition_gamma_greater_ratio}
    \frac{p_2}{p_1} \leq \gamma\,.
\end{align}
Since
\begin{align}
    \frac{p_2}{p_1} \leq \frac{2}{p_1} - 1
\end{align}
condition \eqref{eq::2agent_condition_gamma_greater_ratio} is sufficient for for showing that $\nu_{+} \geq 0 $.

To show that $\nu_{+} \leq 1$, we need that
\begin{align}
    (\gamma +1) p_1 - 2 + \sqrt{\frac{p_1}{p_2}}\Delta &\leq 2(\gamma - 1)
    \\
    \sqrt{\frac{p_1}{p_2}}\Delta &\leq 2 \gamma - (\gamma +1) p_1
\end{align}

We use the same technique as before. The quantity $ 2 \gamma - (\gamma +1) p_1$ is always positive. Squaring both sides gives that we need to meet the condition that
\begin{align}
    \frac{p_1}{p_2} \leq \gamma\,.
\end{align}

Thus we get that if 
\begin{align}
    \max \left\{\frac{p_1}{p_2},\frac{p_2}{p_1} \right\} \leq \gamma
\end{align}
then $0\leq \nu_{+} \leq 1$ and there is an equilibrium point where $\mu_a = \nu_{+}$. For the equilibrium point to be an interior point, we replace the inequality conditions with strict inequality. 

We also need to show that $\mu_b$ is between $0$ and $1$ when these conditions are satisfied. However, we get this by using symmetry to swap opinion $0$ for opinion $1$ (or the same as looking at the dynamics of $\mu_b^1$). In which case, agent $b$ has bias parameter $\gamma >1$ and the same equations hold. 

We then use 
\begin{align}
    \beta_a^* = f(\nu_{+}, \gamma)
\end{align}
to get 
\begin{align}
    \beta_a^* &= \frac{\gamma \left((\gamma + 1) p_1 p_2 - 2 p_2 + \sqrt{p_1 p_2 }\Delta\right) }{(\gamma - 1) ((\gamma + 1) p_1 p_2 + \sqrt{p_1 p_2 }\Delta) }
\end{align}
and by symmetry we get $\beta_b^*$ by using \eqref{eq::reciprocal_for_f}. This gives
\begin{align}
    \beta_b^* &= \frac{2 \gamma p_1 - (\gamma+1)p_1 p_2 - \sqrt{p_1 p_2} \Delta}{(\gamma - 1) ((\gamma + 1) p_1 p_2 + \sqrt{p_1 p_2 }\Delta) }\,.
\end{align}
\end{proof}



\subsubsection{Proof of \Cref{prop::simplified_community_consensus}}

\label{sec::proof_simplified_community_consensus}

\begin{proof}[]
We specifically write this proof for the case when when  $\bbeta(t) \to \bzero$. Converging to $\bone$ follows by symmetry.
To determine if $(\beta_a, \beta_b) = (0,0)$ is a stable equilibrium point, we check the eigenvalues of 
\begin{align}
    \dermat{\bzero} &= 
     \begin{bmatrix}
    \gamma & 0 \\0 &\frac{1}{\gamma} 
    \end{bmatrix}
    \begin{bmatrix}
    p_1 & (1-p_1) \\ (1-p_2) & p_2
    \end{bmatrix}
    \\&= \begin{bmatrix}
   \gamma  p_1  & \gamma(1-p_1) \\ 
    \frac{1}{\gamma}(1-p_2)  & \frac{1}{\gamma}(p_2)
    \end{bmatrix}\,.
\end{align}
The point $(\beta_a, \beta_b) = (0,0)$ is a stable equilibrium if all eigenvalues are less than $1$.

The eigenvalues of $\dermat{\bzero}$ are given by the solution to
\begin{align}
    \lambda^2 - \left(\gamma p_1 + \frac{p_2}{\gamma} \right) \lambda + p_1 + p_2 - 1 = 0\,.
\end{align}
The larger eigenvalue is given by
\begin{align}
    \lambda_{+} = \frac{1}{2}\left(\gamma p_1 + \frac{p_2}{\gamma} \right) + \sqrt{\frac{1}{4}\left(\gamma p_1 + \frac{p_2}{\gamma} \right)^2 + 1 - p_1 - p_2}\,.
\end{align}
To have $\lambda_{+}< 1$, we need
\begin{align}
    \sqrt{\frac{1}{4}\left(\gamma p_1 + \frac{p_2}{\gamma} \right)^2 + 1 - p_1 - p_2} \leq 1 - \frac{1}{2}\left(\gamma p_1 + \frac{p_2}{\gamma} \right)\,.
\end{align}
We square both sides since the right-hand side is always nonnegative. This results in the inequality
\begin{align}
    \gamma p_1 + \frac{p_2}{\gamma} \leq p_1 + p_2
\end{align}
which holds when (if $\gamma > 1$)
\begin{align}
    \gamma \leq \frac{p_2}{p_1}\,.
\end{align}
\end{proof}

\subsubsection{Symmetric Simplified Community Network}

\label{sec::analysis_symmetric_community}

Consider when the community network is symmetric, so let $p = p_1 = p_2$. First, we analyze whether the conditions for consensus (given by  \Cref{thm::eig_does_not_converge}) hold. We need to check the eigenvalues of $\dermat{\bx}$ (defined in \eqref{eq::dermat_def}) for $\bx = \bzero$ and $\bone$. 
For $(\beta_a, \beta_b) = \bzero$,
\begin{align}
    \dermat{\bzero} = \bGamma \bW = 
    \begin{bmatrix}
        \gamma p  & \gamma (1-p) \\
        \frac{1}{\gamma}(1-p) & \frac{1}{\gamma} p  
    \end{bmatrix}
\end{align}
and for $(\beta_a, \beta_b) = \bone$,
\begin{align}
    \dermat{\bone} = \bGamma^{-1} \bW = \begin{bmatrix}
         \frac{1}{\gamma} p & \frac{1}{\gamma}(1-p)  \\
        \gamma (1-p)&  \gamma p 
    \end{bmatrix}\,.
\end{align}




By computing eigenvalues, we find that both $\dermat{\bzero}$ and $\dermat{\bone}$ have an eigenvalue $\lambda_{+}$ where
\begin{align}
    \lambda_{+} 
    & = \frac{p \left(\frac{1}{\gamma} + \gamma \right)  + \frac{1}{\gamma}\sqrt{(\gamma^2 + 1)^2 p^2 + 4\gamma^2 (1-2p)}}{2}
     \\
    & > \frac{2 p  + \sqrt{4 p^2 + 4 (1-2p) }}{2} \\
    & = 1
    \,.
\end{align}
Thus, using \cref{thm::eig_does_not_converge}, the symmetric simplified community network does not converge to a boundary equilibrium point. This is in line with the results of \Cref{prop::simplified_community_consensus}.

Next, we examine $\dermat{\bx}$ for the interior equilibrium point $\bx = (\beta_a, \beta_b)$. Because we assumed this network is symmetric, we take advantage of the relation that $\beta_a = 1 - \beta_b$ and $\mu_a = 1 - \mu_b$ at the equilibrium point. To compute \eqref{eq::dermat_def} for this equilibrium point, we need to compute
\begin{align}
\diag\left(\begin{bmatrix}
     \frac{\gamma_1}{(1+(\gamma_1 - 1)\mu_a)^2}, 
      \frac{\gamma_n}{(1+(\gamma_n - 1)\mu_b)^2}
    \end{bmatrix}\right)\,.
\end{align}
The relation $\mu_a = 1 - \mu_b$ gives that
\begin{align}
    \frac{\gamma}{(1+(\gamma - 1) \mu_a)^2} &= \frac{\gamma}{(1+(\gamma - 1) (1-\mu_b))^2}
    \\ & = \frac{\gamma}{(\gamma - (\gamma - 1)\mu_b)^2}
\end{align}
so we get that at $\bx = (\beta_a, \beta_b)$
\begin{align}
    \dermat{\bx} = 
     \frac{\gamma}{(1+(\gamma - 1) \mu_a)^2}
     \begin{bmatrix}
    p & 1-p \\ 1-p & p
    \end{bmatrix}
    \,.
\end{align}
Ignoring the scalar factor,
the eigenvalues of matrix are $1, -1 + 2p$. 
So long as $ \frac{\gamma}{(1+(\gamma - 1) \mu_a)^2} < 1$, all the eigenvalues are less than $1$. 
This occurs if
\begin{align}
    \mu_a >\frac{\sqrt{\gamma} - 1}{\gamma - 1} = \frac{1}{\sqrt{\gamma} + 1}\label{eq::mua_stable_condition}
    \,.
\end{align}



It turns out that \eqref{eq::mua_stable_condition} is true for all $\gamma > 1$ and all $p \in [0, 1]$. 
Therefore, the matrix $\dermat{\bx}$ has all eigenvalues less than $1$. We conclude that the symmetric simplified community network converges to its interior equilibrium point $\bx$ with probability $1$ (since there are no other interior equilibrium points). 

\fi



\end{document}